\documentclass[letterpaper, 10pt, conference]{ieeeconf}

\IEEEoverridecommandlockouts                              

\usepackage{cite}
\let\labelindent\relax
\usepackage{enumitem}

\usepackage{url}
\usepackage{graphicx}
\usepackage{amssymb}
\usepackage{booktabs}
\usepackage{algorithm, algpseudocode}
\usepackage[switch]{lineno}
\usepackage{amsmath, amsthm}
\usepackage{xcolor}

\newcommand{\reals}{{\mathbb{R}}}

\newcommand{\traj}[1]{\mathbf{#1}}
\newcommand{\trajpi}[1]{\boldsymbol{#1}}

\newcommand{\argmin}{\mathop{\rm argmin}}
\newcommand{\argmax}{\mathop{\rm argmax}}

\newcommand{\mnorm}[1]{{\left\vert\kern-0.25ex\left\vert\kern-0.25ex\left\vert #1 
    \right\vert\kern-0.25ex\right\vert\kern-0.25ex\right\vert}}

\newcommand{\mc}{\mathcal}

\newcommand{\mb}{\mathbb}

\newtheorem{definition}{Definition}

\newtheorem{remark}{Remark}
\newtheorem{proposition}{Proposition}
\newtheorem{assumption}{Assumption}

\newtheorem{problem}{Problem}

\newcommand{\tightdots}{\hbox to 1em{.\hss.\hss.}}

\usepackage{diagbox}
\usepackage[symbol]{footmisc}  
\usepackage{comment}
\usepackage{tikz}
\usepackage[hidelinks]{hyperref}
\allowdisplaybreaks[2]

\begin{document}
\title{When the Correct Model Fails: The Optimality of Stackelberg Equilibria with Follower Intention Updates}
\author{Cayetana Salinas-Rodriguez, Jonathan Rogers, Sarah H.Q. Li
\thanks{
All authors are with the School of Aerospace Engineering at the Georgia Institute of Technology, Atlanta, GA 30332, USA (e-mail: cayesalinas\char64  gatech.edu, jonathan.rogers\char64 ae.gatech.edu, sarahli\char64 gatech.edu).
}%
}

\maketitle
\pagestyle{empty} 

\begin{abstract}
We study a two-player dynamic Stackelberg game where the follower's intention is unknown to the leader. Classical formulations of the Stackelberg equilibrium (SE) assume that the follower's best response (BR) function is known to the leader. However, this is not always true in practice. We study a setting in which the leader receives updated beliefs about the follower BR before the end of the game, such that the update prompts the leader and subsequently the follower to re-optimize their strategies. We characterize the optimality guarantees of the SE solutions under this belief update for both open loop and feedback information structures. Interestingly, we prove that in general, assuming an incorrect follower's BR may lead to a lower leader cost over the entire game than knowing the true follower's BR. We support these results with numerical examples in a linear quadratic (LQ) Stackelberg game, and use Monte Carlo simulations to show that the instances of incorrect BR achieving lower leader costs are non-trivial in collision avoidance LQ Stackelberg games. 
\end{abstract}

\section{Introduction}\label{sec:intro}
Autonomous systems have become increasingly prevalent in society, whether in the form of self-driving cars, delivery robots, or humanoids replacing household workers. As they are commonly deployed in environments shared by humans, vehicles, and other robots, their ability to \emph{interact} effectively with others becomes more critical to overall performance. In these scenarios, the system must learn the interaction model and optimize its strategy simultaneously. Active research areas such as intent estimation and inverse learning~\cite{ng_algorithms_2000, markakis_inverse_2015, bayesianIRL} \emph{implicitly} assume that learning the correct interaction model and solving for the certainty-equivalence strategy is the correct approach. However, we re-examine the question:

\emph{does knowing the correct interaction model lead to the optimal interaction strategy?}

If autonomous systems do not interact, i.e., if at most one autonomous system bases its strategy on the decisions of others, then that system can treat the others as dynamic obstacles and perform obstacle avoidance by estimating their dynamics online. In this case, our research question becomes trivial: knowing the true model leads to better interaction outcomes~\cite{optimal-filtering}. However, decision-making systems inherently interact in shared environments, and not modeling their interaction leads to incorrect predictions and unsafe outcomes.

To address this question, this work considers a two-player non-cooperative Stackelberg game~\cite{von_stackelberg_market_1934,basar_dynamic_nodate} where the leader estimates the follower's BR and optimizes its Stackelberg strategy under the current best BR guess simultaneously. We study both open loop (OL) and feedback (FB) information structures with linear time invariant (LTI) dynamics and a finite, discrete time horizon. We focus on strategy optimization and assume the estimated follower's BR comes from an oracle that may update the BR estimate during the control horizon.

\textbf{Contributions}. For OL Stackelberg equilibria, we prove a sufficient condition for when optimizing with respect to a correct BR belief leads to the optimal leader strategy (implying that the correct BR belief does not always lead to an optimal leader strategy), and demonstrate its connection to the time-inconsistency property of SE. For FB Stackelberg equilibria, we show that not all solutions imply a correct BR leads to the optimal leader strategy despite its time-consistency. Finally, numerical simulations of an LQ Stackelberg game with Bayesian updates of BR estimation demonstrate that these theoretical results frequently occur and lead to sub-optimal behavior when not accounted for. This work opens broader questions, such as whether the cumulative cost is the right metric for evaluating these problems, and whether accurately learning the intentions of other autonomous systems is always worth the effort.

 \section{Related Research}\label{sec:related_res}
First introduced in~\cite{von_stackelberg_market_1934} to model oligopolistic markets, the Stackelberg game and its dynamic extension~\cite{chen_stackelberg_1972,simaan_stackelberg_nodate,si_additional_nodate} are used extensively in control theory to model cyber-security threats~\cite{kar2017trends}, lane-merging scenarios~\cite{shi2025stackelberg}, and power-grid optimization~\cite{yu_real-time_2016}. An important consideration in Stackelberg games is the information structure; common structures include OL~\cite{kydland_noncooperative_1975}, FB~\cite{gardner_feedback_1977}, and closed loop (CL)~\cite{basar_closed-loop_1979}. Both OL and FB Stackelberg equilibria have been analytically solved for the discrete-time LQ game~\cite{gardner_feedback_1977, kydland_equilibrium_1977}.

Although computing the SE assumes that the leader knows the follower's BR, this is not always true in practice. Estimating opponents' intentions is the focus of inverse learning~\cite{markakis_inverse_2015} in interactive settings:~\cite{sadigh2016planning} uses inverse reinforcement learning~\cite{ng_algorithms_2000} to model autonomous car-human interaction,~\cite{Hoermann2017probabilisticPredictions} uses probabilistic motion planning to model autonomous vehicles, and~\cite{geiger2021learning} uses neural network predictions in multi-agent trajectory Nash games. Unknown opponent intention for Nash games is addressed in~\cite{peters_inferring_2021, liu_learning_2023, LUCIDGames}, while follower intent uncertainty in Stackelberg games is studied in~\cite{lauffer_no-regret_2024, balcan2015commitment, bollini2026learningbayesianstackelberggames} from an online learning and repeated game perspective. In contrast, this paper studies the one-shot dynamic Stackelberg game where the belief update occurs within the finite horizon, and explicitly challenges the assumption that the true belief yields the lowest cost by addressing the optimality conditions of the known SE under updated beliefs.

\section{The Stackelberg game under unknown follower intentions} \label{sec:problem_definition}
\textbf{Notation}
We use $\reals^n$ to denote a real-valued vector space of dimension $n$, $\reals^{n\times m}$ to denote a real-valued matrix with dimensions $n\times m$, and $\mathbb{N}$ to denote the set of all natural numbers. The discrete number set $\{0,\ldots, T\}$ is denoted as $[T]$ and the discrete number set $\{t_0,\ldots, t_f\}$ is denoted as $[t_0, t_f]$. 
Additionally, we denote the state trajectory as $\traj{x} = (x_{0}, \ldots, x_{T}) \in \reals^{n(T+1)}$, the control trajectory as $\traj{u} = (u_{0}, \ldots, u_{T-1}) \in \reals^{mT}$, and the strategy trajectory as $\trajpi{\pi} = (\pi_{0}, \ldots, \pi_{T-1}) \in \reals^{(m\times n)T}$.
Finally, we denote by $\mb{S}^{n}_{+}$ the space of $n\times n$ symmetric positive semi-definite matrices and $\mb{S}^{n}_{++}$ the corresponding space of positive definite matrices.
\subsection{Stackelberg dynamic game}\label{sec:stackelberg_game_definition}
We consider two players, a leader $L$ and a follower $F$ over a finite and discrete time horizon $[T]$. The players jointly follow deterministic, LTI system dynamics given by
\begin{align}
    {x}_{t+1} &= A {x}_t + B^L u^L_t + B^F u^F_t,\quad \forall \ t \in  [T-1],
    \label{eq:dynamics}
\end{align}
where ${x}_t \in \reals^{n}$ is the shared time-dependent system state. For player $ i \in \{L,F\} $,  $B^i \in \reals^{n\times m_i},$ is its control matrix, and  $u^i_t\in \reals^{m_i}$ is its time-varying control. From initial state $x_0\in\reals^n$ and under control sequences $(\traj{u}^L,\traj{u}^F)$, we denote the feasible state trajectory from~\eqref{eq:dynamics} as
\begin{equation}
\begin{aligned}
    \mathbf{x} = Hx_0 + G^L \traj{u}^L + G^F \traj{u}^F \in \reals^{n(T+1)},
     \label{eq:state_trajectory}
\end{aligned}
\end{equation}
where the matrices $H \in \reals^{n(T+1)\times n}$, $G^i \in \reals^{n(T+1) \times m_iT}$ for $i\in\{L,F\}$ are defined in~\eqref{eq:H_G_matrix}.
For a time horizon $[c, d]$ with horizon length $\hat{T} = d-c$, we use $H_{c:d}\in \reals^{n(\hat{T} + 1)\times n}$, $G^i_{c:d} \in \reals^{n(\hat{T} + 1)\times m_i\hat{T}}$ to denote the state trajectory in horizon $[c,d]$ from state $x_c$ as
\begin{equation}
\begin{aligned}
    x_{c:d} = H_{c:d}x_c + G^L_{c:d} {u}_{c:d-1}^L + G^F {u}_{c:d-1}^F \in \reals^{n(\hat{T} + 1)}.
     \label{eq:truncated_state_traj}
\end{aligned}
\end{equation}
\textbf{Information Structure}.
Players select strategies to generate control based on the available information at each time step $t$, which we denote by $\mc{I}^i_t$ for $i \in \{L, F\}$. We consider two information structures: 
\begin{enumerate}
    \item open loop information $\mathcal{I}^i_t = \{ x_0 \}, \forall t \in [T-1]$;
    \item feedback information $\mathcal{I}^i_t = \{ x_t\}, \ \forall t \in [T-1]$.
\end{enumerate}
Each information structure leads to a distinct strategy set~\cite{basar_dynamic_nodate}. We denote player $i$'s strategy set by $\Pi^i_t$. The strategy sets corresponding to  the two information structures above are explicitly given by
\begin{equation}\label{eq:information_structures}
    \Pi^i_t = \begin{cases}
         \mc{U}^i_t = \{u^i_t \in \reals^{m_i}\},& \mathcal{I}^i_t = \{ x_0 \}\\
         \mc{K}^i_t = \{K^i_t \in \reals^{m_i \times n}\}, & \mathcal{I}^i_t = \{ x_t\}
    \end{cases}, i \in \{L,F\}.
\end{equation}
Under FB information structure, the  control chosen is $u^i_t~=~- K^i_t x_t$ for $i \in \{L,F\}$. 
We use $\Pi^i_{0:T-1}$ to denote player $i$'s strategy domain over the finite horizon $[T-1]$, and $\trajpi{\pi}^i\in \Pi^i_{0:T-1}$ to denote a specific strategy for player $i$. When the information structure context \(\mc{I}^i\) is clear, we use $\trajpi{\pi}^i$ to denote \(\trajpi{\pi}^i(\mc{I}^i)\). 

\noindent\textbf{Player Objectives}. At time step $t \in [T-1]$, each player $i \in \{L,F\}$ has a state and control-dependent cost function $J^i_t: \reals^{n\times m_L\times m_F} \mapsto \reals$. At time step $T$, players have a state-dependent cost function $J^i_T: \reals^{n} \mapsto \reals$.
Each player's total cost over time horizon $[T]$ is given by 
\begin{equation}
    J^{i}_{0:T}\big(\traj{x}, \traj{u}^L, \traj{u}^F\big)= J^i_T\big(x_T\big)\\+\sum_{t=0}^{T-1} J^i_t\big(x_t, u^L_t, u^F_t\big).
    \label{eq:cost_definition}    
\end{equation}
\begin{assumption}\label{assum:convexity}
The leader and follower objectives  $J^{L}_{0:T}$ and $ J^{F}_{0:T}$ are continuously differentiable and strictly convex in $\big(\traj{u}^L,  \traj{u}^F\big)$, and continuously differentiable and convex in $\traj{x}$.
\end{assumption}
\noindent\textbf{Best response (BR)}.  Under a leader's strategy $\trajpi{\pi}^L$, the follower seeks to minimize its objective $J^F_{0:T}$~\eqref{eq:cost_definition} using its strategy $\trajpi{\pi}^F$. 
This problem is formulated as
\begin{equation}
\begin{aligned}
\label{eq:follower_problem}
\min_{\trajpi{\pi}^F\in\Pi^F_{0:T-1}}& \quad  J^F_{0:T}\big(\traj{x}, \trajpi{\pi}^L(\mc{I}^L), \trajpi{\pi}^F(\mc{I}^F)\big)\\
\text{s.t.}\quad 
& \traj{x} =  Hx_0 + G^L \trajpi{\pi}^L + G^F \trajpi{\pi}^F.
    \end{aligned}
\end{equation}
The follower's BR is the $\argmin$ set of~\eqref{eq:follower_problem}. Each leader strategy maps to a set of optimal  follower responses. Under Assumption~\ref{assum:convexity}, this mapping is unique and continuous in the leader strategy $\trajpi{\pi}^L$~\cite[Sec. 7.2]{basar_dynamic_nodate} and can be parametrized by the follower's objective function. We denote this continuous function as $b^\star:\Pi_{0:T-1}^L\mapsto \Pi_{0:T-1}^F$, so that $b^\star(\trajpi{\pi}^L)$ is the optimal follower strategy under~\eqref{eq:follower_problem} for all $\trajpi{\pi}^L \in \Pi_{0:T-1}^L$. Additionally, we use $\mc{F}$ to denote the set of all continuous follower response functions, given by
\begin{equation}\label{eq:all_responses}
    \mc{F} = \{b \ | \ b:\Pi_{0:T-1}^L\mapsto \Pi_{0:T-1}^F\}.
\end{equation}
Next, we define the Stackelberg dynamic game. 
\begin{definition}[Stackelberg Dynamic Game]\label{def:stackelberg_game}
In a two player Stackelberg dynamic game with leader $L$ and follower $F$, the leader seeks a strategy $\trajpi{\pi}^L$ that solves
\begin{equation}
 \begin{aligned}
\label{eq:leader_problem}
\min_{\trajpi{\pi}^L\in \Pi^L_{0:T-1}} \quad &  J^{L}_{0:T}\big(\traj{x}, \trajpi{\pi}^L(\mc{I}^L), b^\star(\trajpi{\pi}^L)\big)\\
\text{s.t.}\quad 
& \traj{x} = Hx_0 + G^L \trajpi{\pi}^L + G^F b^\star(\trajpi{\pi}^L),
\end{aligned}   
\end{equation}
while the follower seeks a strategy $\trajpi{\pi}^F$ that solves~\eqref{eq:follower_problem}. In particular, the leader always announces $\trajpi{\pi}^L$ first and the follower responds with $\trajpi{\pi}^F$.
\end{definition}
Under OL and FB information structures and with known follower's BR, joint strategies that solve~\eqref{eq:leader_problem} exist and lead to different mathematical guarantees~\cite{basar_dynamic_nodate}.
We discuss each specific guarantee in Sections~\ref{sec:optimality_conditions_OLSE_intention} and~\ref{sec:feedback_opt_break}, respectively. 
\subsection{Stackelberg game under follower BR belief updates} \label{sec:ambiguous_stackelberg_game_definition}
In this subsection, we formulate a Stackelberg dynamic game in which the leader does not know the follower's BR $b^\star$ but has time-varying beliefs about it. 
 \begin{assumption}[Follower's BR]\label{ass:BR_belief}
 Over the time horizon $[T]$, the follower has a time-invariant BR given by
 \begin{equation}
     b^\star \in \mc{F}.
 \end{equation}
 The leader has two beliefs about the follower's BR: at $t=0$, it is $b^1\in\mc{F}$, and at an \textbf{update time} $\tau \in (0,T]$, it is $b^2 \in \mc{F}$.  
 \end{assumption}
 The specific method used to retrieve the follower's beliefs is less important here. However, many methods from estimation and machine learning can  predict the follower's BR under different observation inputs \cite{ng_algorithms_2000,neu_apprenticeship_2012,markakis_inverse_2015,bayesianIRL}. 

\begin{figure}[H]
    \centering
    \includegraphics[width=1
    \linewidth]{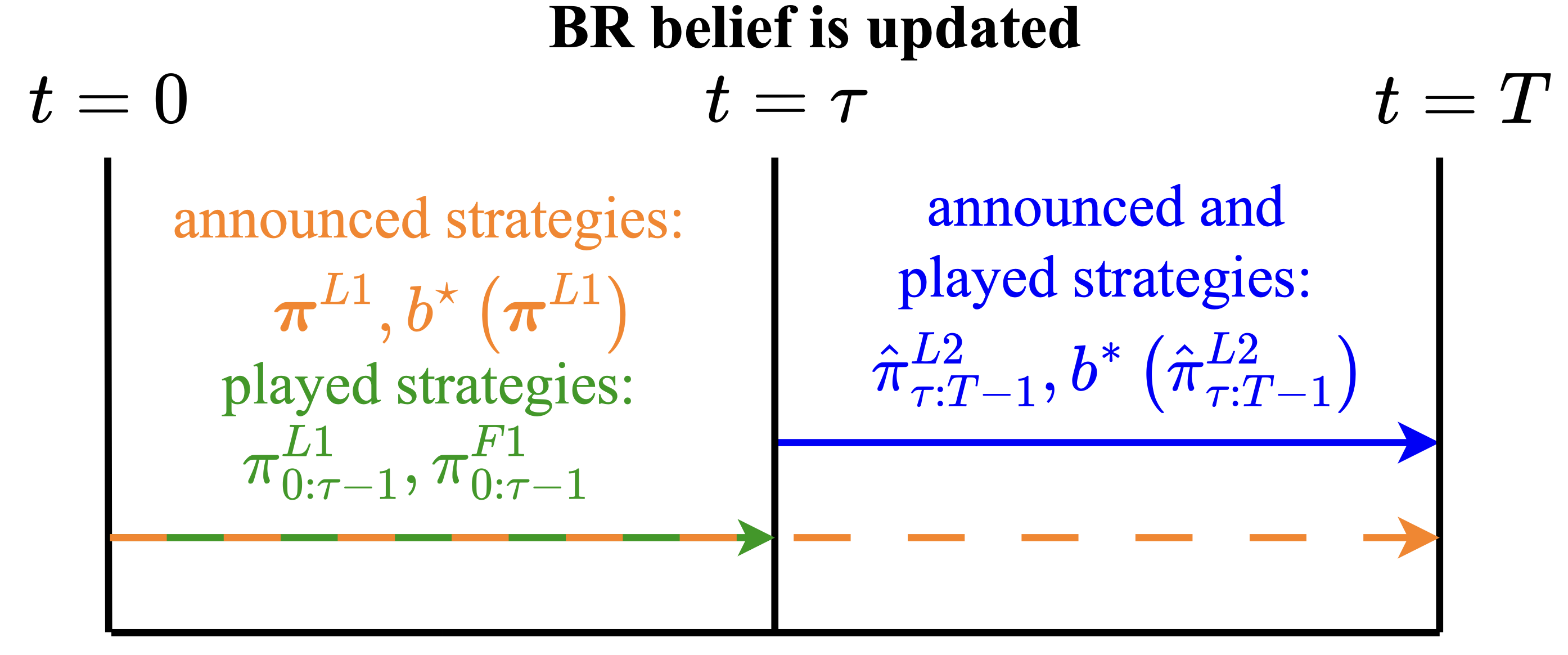}
    \caption{Dynamics of Stackelberg game with BR update.}
    \label{fig:scenario_I_diagram}
\end{figure} 
\noindent\textbf{Stackelberg game with BR belief update}. Under Assumption~\ref{ass:BR_belief}, consider a Stackelberg game in which the leader does not know the follower's true BR $b^\star\in\mc{F}$. In this setting, the leader first solves the $[T]$-horizon Stackelberg game~\eqref{eq:leader_problem} under belief $b^1\in\mc{F}$ to compute a strategy $\trajpi{\pi}^{L1} \in \Pi^L_{0:T-1}$ and the follower solves its $[T]$-horizon problem~\eqref{eq:follower_problem} under its true BR $b^\star$ to compute a strategy $\trajpi{\pi}^{F1} = b^\star(\trajpi{\pi}^{L1})$. For $\tau \in (0,T]$ time steps, the system dynamics~\eqref{eq:truncated_state_traj} evolve as
\begin{equation}\label{eq:pre_update_states}
    x^1_{0:\tau} =H_{0:\tau}x_0 + G^L_{0:\tau} \pi^{L1}_{0:\tau-1} + G^F_{0:\tau} \pi^{F1}_{0:\tau-1}.
\end{equation}
At the update time $\tau$, the system state arrives at state $x_{\tau}$. The leader receives an updated BR belief $b^2\in \mc{F}$, re-solves the Stackelberg game under $b^2$ for the remaining horizon $[\tau,T]$, and executes a new optimal strategy $\hat{\pi}^{L2}_{\tau:T-1}$ that minimizes  
\begin{equation}\label{eq:truncated_problem}
\begin{aligned}
&\min_{\pi^L_{\tau:T-1}\in\Pi^L_{\tau:T-1}}  J^L_{\tau:T}\left(x_{\tau:T},\pi^L_{\tau:T-1},b^2(\pi^L_{\tau:T-1})\right)\\
&\text{s.t. }
\quad x_{\tau:T} = H_{\tau:T} x^1_{\tau}+ G^L_{\tau:T} \pi^L_{\tau:T-1} + G^F_{\tau:T} b^2(\pi^L_{\tau:T-1}).\\
\end{aligned}
\end{equation}
Once the leader announces new strategy $\hat{\pi}^{L2}_{\tau:T-1}$, the follower also updates its best response  to $b^\star(\hat{\pi}^{L2}_{\tau:T-1})$, and the system dynamics evolve from $x_\tau$ as 
\begin{equation}\label{eq:post_update_states}
    x^2_{\tau:T} =H_{\tau:T}x^1_{\tau} + G^L_{\tau:T} \hat{\pi}^{L2}_{\tau:T-1} + G^F_{\tau:T} b^\star(\hat{\pi}^{L2}_{\tau:T-1}).
\end{equation}
Over the time horizon $[T]$ and under the BR belief update, the leader's composite strategy is given by
\begin{equation}
    \hat{\trajpi{\pi}}^L = \left[\pi^{L1}_{0:\tau-1}, \hat{\pi}^{L2}_{\tau:T-1}\right],
    \label{eq:unknown_BR_strategy_leader}
\end{equation}
the follower's composite strategy given by
\begin{equation}
    \hat{\trajpi{\pi}}^F = [\pi^{F1}_{0:\tau-1}, b^\star(\hat{\pi}^{L2}_{\tau:T-1})],
\end{equation}
and the composite state  trajectory is given by
 \begin{equation}
     \hat{\traj{x}}=[x^1_{0:\tau-1}, x^2_{\tau:T}]\in \reals^{n(T+1)}.
     \label{eq:traj_BR_update}
 \end{equation}
To highlight how the leader's total cost depends on the change in its BR belief, we break down the leader's cost~\eqref{eq:cost_definition} into the components $J^L_{0:\tau-1}(b^1)$ and $J^L_{{\tau:T}}(b^2)$, defined as
\begin{align}
    J^L_{0:\tau-1}(b^1) &= \sum_{t=0}^{\tau-1}J^L_t\big(x^1_t, \pi^{L1}_t, b^\star({{\pi}}_{t}^{L1})\big), \label{eq:cost_decomp_part1}\\
    J^L_{{\tau:T}}(b^2) &= J_T^L(x^2_T) + \sum_{t=\tau}^{T-1}J^L_{t}\big({x}^2_t, \hat{{\pi}}_t^{L2}, b^\star(\hat{{\pi}}_t^{L2}) \big),
    \label{eq:cost_decomp_part2}
\end{align}             
such that $J^L_{0:T}(\hat{\traj{x}}, \hat{\trajpi{\pi}}^L, \hat{\trajpi{\pi}}^F) = J^L_{0:\tau-1}(b^1) + J^L_{{\tau:T}}(b^2)$.

This BR belief update setting captures applications of Stackelberg games in iterative or data-driven settings such as~\cite{zhao2024cooperation,zhang2025motion}. However, it differs from active inverse learning game-theoretic approaches in which the control is strictly selected to uncover follower BR~\cite{ward_active_2024,wu2022inverse}. 

The key question we explore in this manuscript is whether the leader’s cost is the lowest when it knows the true BR $b^\star$ versus all possible BR beliefs $(b^1, b^2) \in \mc{F}$. 

\begin{problem}[Stackelberg game under BR belief update]\label{pro:change_br}
We consider a Stackelberg game in which the leader updates its belief about the follower’s BR at time $\tau$ and the follower always utilizes BR $b^\star \in \mc{F}$. Does the true BR belief always lead to the lowest cost?   I.e.,  is it true that for all beliefs 
\begin{multline}
J^L_{0:\tau-1}(b^\star) + 
      J^L_{{\tau:T}}(b^\star) 
      \leq 
      J^L_{0:\tau-1}(b^1) + 
      J^L_{{\tau:T}}(b^2), \\ \forall \ b^1, b^2\in\mc{F}.
      \label{eq:question_1_paper}    
\end{multline}
\end{problem} 
We investigate Problem~\ref{pro:change_br} for both the OL and FB information structures~\eqref{eq:information_structures}.  

\textbf{Motivating example.} An autonomous food delivery robot must complete an order while navigating around people along its path. Because its actions influence how people move, the robot models their intentions as part of its strategy optimization. But does correctly estimating those intentions always yield the optimal control strategy? After discussing the answer to this question in Section~\ref{sec:issues_solutions}, we provide simulation results for a similar scenario in Section~\ref{sec:simulations}.

\section{SE under BR Updates} \label{sec:issues_solutions}
This section addresses Problem~\ref{pro:change_br}
under OL and FB information structures~\eqref{eq:information_structures}. We first define the optimality conditions of both SE solutions, followed by a discussion on whether the SE solutions with BR belief updates under the true BR belief satisfy~\eqref{eq:question_1_paper}.   
\subsection{OL Stackelberg equilibrium (OLSE) under BR updates}
\label{sec:optimality_conditions_OLSE_intention}
We first evaluate the OLSE guarantee~\cite[Sec.7.2]{basar_dynamic_nodate} under BR belief updates.
When the leader solves the Stackelberg game~\eqref{eq:leader_problem} from state $x_0 \in \reals^{n}$ against a follower whose BR is $b^\star \in \mc{F}$, the leader selects the control sequence $\traj{u}^{L\star}$ that satisfies
\begin{equation}
\begin{aligned}
        \traj{u}^{L\star} & \in \argmin_{\traj{u}^{L}\in \mc{U}^L_{0:T-1}} J^L_{0:T}\left(\traj{x}, \traj{u}^{L}, b^\star\big(\traj{u}^{L}\big)\right)\\
        \text{s.t. } & \traj{x} = Hx_0+ G^L\traj{u}^{L} + G^F b^\star\big(\traj{u}^{L}\big).    
\end{aligned}
\label{eq:optCond_OL_intention}    
\end{equation} 
The optimality condition of the OLSE implies that the leader's OL control $\traj{u}^{L\star}$ is optimal strictly at initial state $x_0$,  with respect to all control sequences $\traj{u}^L \in \reals^{m_LT}$ 
announced at $t=0$, and executed over the entire horizon $[T]$.

An important property of certain SE is the so-called time-inconsistency, which we now define.
\begin{definition}[SE Time-inconsistency~\cite{basar_time_1989}] \label{def:time_inconsistency}
Consider the SE $\left(\trajpi{\pi}^{L\star}, \trajpi{\pi}^{F\star}\right)$ that solves 
\eqref{eq:follower_problem} and \eqref{eq:leader_problem} from initial state $x_0 \in \reals^n$ with information structures $(\mc{I}^L, \mc{I}^F)$.
If at time step $t\in [T]$ and state $x_t \in \reals^n$ under dynamics~\eqref{eq:state_trajectory},
we re-solve \eqref{eq:follower_problem} and \eqref{eq:leader_problem} for the remaining horizon and obtain SE $\left(\hat{\pi}^{L}_{t:T-1}, \hat{\pi}^{F}_{t:T-1}\right)$.
We say that the SE $\left(\trajpi{\pi}^{L\star}, \trajpi{\pi}^{F\star} \right)$ is \textbf{time-inconsistent} if
\begin{equation}
    \pi^{L\star}_{t:T-1} \neq \hat{\pi}^L_{t:T-1}.
\end{equation}
If the SE is not time-inconsistent, then it is \textbf{time-consistent}.
\end{definition}
The OLSE is known to be time-inconsistent~\cite{basar_time_1989}, meaning that we can separately analyze the pre-update and post-update horizons. Over $[0,\tau-1]$, we first consider  whether $J^L_{0:\tau-1}(b^\star) \leq J^L_{0:\tau-1}(b^1)$ for all $b^1 \in \mc{F}$. From~\eqref{eq:optCond_OL_intention} and pre-update state dynamics~\eqref{eq:pre_update_states}, the leader executes $u^{L\star}_{0:\tau-1}$ up to time step $\tau$. In this truncated horizon, $u^{L\star}_{0:\tau-1}$ is in general \emph{not} optimal for $J^L_{0:\tau-1}$~\eqref{eq:cost_decomp_part1}. In fact, the optimal leader's strategy over $[0,\tau-1]$ is one that minimizes the truncated leader cost under true follower BR, $J^L_{0:\tau-1}(b^\star)$, which is different from the $[0,\tau-1]$ component of $\traj{u}^{L\star}$~\eqref{eq:optCond_OL_intention}, the minimizer of the full horizon leader cost under true follower BR, $J^L_{0:T}(b^\star)$. In particular, this implies that an alternative belief $b^1 \in \mc{F}$, when solved over horizon $[T]$, may lead to the control sequence $u^{L1}_{0:\tau-1}$ that minimizes $J^L_{0:\tau-1}(b^\star)$. \ref{app:example1} shows an instance where this situation arises in a LQ Stackelberg game.

Over the time horizon $[\tau, T]$,  we derive a sufficient condition for when the true BR belief minimizes the leader's cost $J^L_{{\tau:T}}$~\eqref{eq:cost_decomp_part2}. After the BR update at $\tau$, we compare the leader’s cost when it assumes a belief $b^2 \in \mc{F}$ to when it re-optimizes under the true BR belief $b^\star$. 

\begin{proposition}\label{prop:part2_OLSE_prob1}
Consider the state trajectories $x^1_{0:\tau}$ and $x^\star_{0:\tau}$~\eqref{eq:pre_update_states} when the leader has BR belief $b^1$ and $b^\star$, respectively, prior to BR update at $\tau < T$ in the OL information setting~\eqref{eq:information_structures}. If the states $x^1_{\tau} = x^\star_{\tau}$. Then,
 \begin{equation}
    J^L_{\tau:T}(b^\star) \leq J^L_{\tau:T}(b^2), \quad \forall b^2 \in \mc{F},
    \label{eq:opt_OLSE_prob1_second_half}
\end{equation}
where $b^2$ is the follower BR belief after time $\tau$.
\end{proposition}
\begin{proof}
The result follows directly from the optimality condition of OLSE~\eqref{eq:optCond_OL_intention} 
for the post-BR-update leader problem~\eqref{eq:truncated_problem}.
Since $\hat{u}^{L\star}_{\tau:T-1}$~\eqref{eq:unknown_BR_strategy_leader} minimizes the leader’s cost $J^L_{\tau:T}$ under the true BR $b^\star$ for the new ``initial" state $x^\star_\tau \in \reals^n$. Then, for any other $\hat{u}^{L}_{\tau:T-1} \in \reals^{m_L(T-\tau)}$ announced at $\tau$ and state $x^\star_\tau$, the cost will be greater. Since the belief $b^2 \in \mc{F}$ induces a corresponding control sequence $\hat{u}^{L2}_{\tau:T-1}$, 
this yields~\eqref{eq:opt_OLSE_prob1_second_half}.
\end{proof}
\begin{remark}
    Proposition~\ref{prop:part2_OLSE_prob1} highlights that the wrong BR belief may obtain a lower cost over $[\tau, T]$. The condition $x^1_{\tau} = x^\star_{\tau}$ guarantees that the true follower's BR belief $b^\star$ minimizes the remaining cost. However, it is quite restrictive and in general false for a belief $b^1 \neq b^\star$. Thus, a wrong belief may obtain a lower cost after the update.
\end{remark} 

We have shown that over the horizon $[0,\tau-1]$ there may be a BR belief that leads to better leader cost. Then, by Proposition~\ref{prop:part2_OLSE_prob1}, there may also be a better BR belief on $[\tau,T]$. Due to the time-inconsistency of the OLSE (Def.~\ref{def:time_inconsistency}), there may be a joint pair $(b^1, b^2)\in\mc{F}$ that achieve lower leader cost over $[T]$ when compared to always assuming true BR $b^\star$. Thus, in general, \eqref{eq:question_1_paper} will not hold for all $(b^1, b^2) \in \mc{F}$. We give a numerical example where a wrong set of beliefs $(b^1, b^2)$ obtain a lower cost than the true BR belief $b^\star$ in a LQ game in~\ref{app:example1}.

\subsection{FB Stackelberg equilibrium (FSE) under BR updates}\label{sec:feedback_opt_break}
This section considers whether FSE under follower BR updates satisfies \eqref{eq:question_1_paper}. We first adapt the FSE optimality condition~\cite[Sec.3.6]{basar_dynamic_nodate} under BR updates. When the leader solves the Stackelberg game~\eqref{eq:leader_problem} against a follower whose BR is $b^\star \in \mc{F}$, the leader selects the control FB gain $K^{L\star}_t$ that satisfies
\begin{multline}
        K^{L\star}_t \in \argmin_{K^L_t \in \mc{K}^L_t} J^L_{t:T}(x_{t:T}, K^L_t, K^{L\star}_{t+1:T-1},\\ b^\star(K^L_t), K^{F\star}_{t+1:T-1}), \quad\forall t\in [T],
        \label{eq:feedback_optimality_ambiguous}
\end{multline}
and $K^{F\star}_t = b^\star\left(K^{L\star}_t\right)$ is the follower's BR~\eqref{eq:follower_problem}. The optimality condition~\eqref{eq:feedback_optimality_ambiguous} implies that the FB strategy at time $t$ is optimal against leader strategies where only $K^{L}_t$ can be changed and $K^{L}_{\hat{t}}$ is fixed for time steps $\hat{t} > t$. This optimality definition ensures that, unlike the OLSE, FSE is time-consistent~\cite{basar_time_1989}. However, we will show that this is still not enough to affirmatively answer Problem~\ref{pro:change_br}.\\
We now define another property of the FSE that will be relevant in the discussion of whether assuming the true BR will lead to obtaining the lowest cost.
\begin{definition}[Markov Perfect FSE (MPFSE)~\cite{fudenberg1991game}]
    A joint FB strategy $\big( \traj{K}^{L\star}, \traj{K}^{F\star}\big)$ is a MPFSE if for all $t \in [T-1]$, in addition to satisfying~\eqref{eq:feedback_optimality_ambiguous}, the leader's strategy $K^{L\star}_{t:T-1}$ satisfies
    \begin{multline} \label{eq:MPFSE}
        J^L_{t:T}\Big(x_{t:T},K^{L\star}_{t:T-1}, b^\star( K^{L\star}_{t:T-1})\Big) \\\leq  J^L_{t:T}\Big(x_{t:T}, K^{L}_{t:T-1}, b^\star(K^{L}_{t:T-1})\Big), \ \forall   K^{L}_{t:T-1} \in \mc{K}^L_{t:T-1}.
    \end{multline}
\end{definition}
MPFSE is a stronger condition than FSE, which only implies time-consistency~\cite{basar_time_1989}. In addition to ensuring that the the joint strategy is consistent when resolved in future time steps, MPFSE also ensures that $K^{L\star}_{t:T-1}$ is optimal in the set of all leader FB strategies from time $t$ to $T-1$~\cite[Chp.13.2]{fudenberg1991game}.

Markov perfect Nash equilibrium exists for certain subsets of LQ dynamic games~\cite{maskin1987theory,maskin1988theory}, while MPFSE's existence is less well explored. In~\cite{goktas2022zero}, MPFSE is shown  to exist for two player zero-sum Stackelberg games. 

Since the FSE is time-consistent, when the belief over the horizon remains constant, $b^1 = b^2 = b$, the sequence of gains $\hat{\traj{K}}^{L}$\eqref{eq:unknown_BR_strategy_leader} will be the same as $\traj{K}^{L}$ which solve \eqref{eq:feedback_optimality_ambiguous} under the belief $b \in \mc{F}$. 
\begin{proposition}\label{prop:feedback_prob1}
    Consider the BR update~\eqref{eq:unknown_BR_strategy_leader} with FB information structure~\eqref{eq:information_structures}. If the FSE is not Markov perfect, there exist a BR $(b^1, b^2)\in \mc{F}$, $(b^1, b^2) \neq b^\star$ that yields a lower leader's cost $J^L_{0:T}$~\eqref{eq:cost_definition} than the true belief $b^\star\in\mc{F}$, i.e., 
    \begin{equation}
      J^L_{0:\tau-1}(b^1) + 
      J^L_{{\tau:T}}(b^2)
      <
       J^L_{0:\tau-1}(b^\star) + 
      J^L_{{\tau:T}}(b^\star).
      \label{eq:other_Model_better_FB}
    \end{equation}
\end{proposition}
\begin{proof}
    If the FSE is not Markov perfect, this implies that there exists some sequence of FB gains $\traj{K}^{L'}\neq\traj{K}^{L\star}$ such that they minimize $J^L_{0:T}$~\eqref{eq:cost_definition} over all other possible gain sequences with follower BR $K^{F'}_t = b^\star(K^{L'}_t)$. We construct a response function $b^1\in \mc{F}$ such that the solution to \eqref{eq:leader_problem} under response $b^1$ is denoted $\left(\traj{K}^{L1}, \traj{K}^{F1}\right)$ and $K^{F1}_{t} = b^\star(K^{F1}_{t})$ for $t \in [0,\tau-1]$. We assume that, in addition, this $b^1$ satisfies
    \begin{equation}
        K^{L1}_{0:\tau-1} = K^{L'}_{0:\tau-1} \text{ and } K^{F1}_{0:\tau-1} = K^{F'}_{0:\tau-1}.
    \end{equation}
    Similarly, we can define another $b^2\in \mc{F}$ such that the solution to \eqref{eq:truncated_problem} under belief $b^2$ is $\left(\hat{K}^{L2}_{\tau:T-1}, \hat{K}^{F2}_{\tau:T-1}\right)$ and $\hat{K}^{F2}_{t} = b^\star(\hat{K}^{F2}_{t})$. In addition, this $b^2$ satisfies
    \begin{equation}
        \hat{K}^{L2}_{\tau:T-1} = K^{L'}_{\tau:T-1} \text{ and } \hat{K}^{F2}_{\tau:T-1} = K^{F'}_{\tau:T-1}.
    \end{equation}
    The fact that the FSE is time-consistent (Def.~\ref{def:time_inconsistency}) means that these two beliefs can be combined into a joint pair $(b^1, b^2)\in\mc{F}$ that obtains lower cost than the true belief $b^\star$, implying~\eqref{eq:other_Model_better_FB}.
\end{proof}
Proposition~\ref{prop:feedback_prob1} implies that for FSEs in general, some follower BR beliefs other than the true BR may yield a lower leader cost. Since the true follower's BR FSE does not minimize the $[T]$ horizon cost over all possible gain sequences, there may be multiple pairs $(b^1, b^2)\in\mc{F}$ that obtain lower cost than $b^\star$. \ref{app:example2} provides an example where the wrong BR belief obtains a lower cost than the true BR belief FSE. 

Finally, we briefly discuss the closed loop Stackelberg equilibrium (CLSE)~\cite{basar_dynamic_nodate}. CLSE uses the information structure $\mc{I}^i_t = \{x_0, \ldots, x_t\}$ for all $t$ and $i\in\{L,F\}$, and provides stronger optimality guarantees than both FSE and OLSE~\cite{kydland_equilibrium_1977}. However, the CLSE is more complex to compute, even in LQ games, and is generally time-inconsistent~\cite{basar_time_1989}, except when the leader can enforce a team-solution on the follower. While this paper does not study the CLSE, we would like to point out that despite the stronger guarantees, the CLSE may suffer from similar optimality issues as the OLSE in the Stackelberg game with BR updates.

\section{Simulation results}\label{sec:simulations}
We construct a Stackelberg LQ game that represents a trajectory collision avoidance interaction, such as a human-robot encounter where limited communication is available, for which we verify our proven results from Section~\ref{sec:issues_solutions}.
\subsection{Simulation set up}
We treat each player as two dimensional double integrator with dynamics
\begin{equation}
    x^i_{t+1} = \begin{bmatrix}
        1 & 1\\ 0 & 1
    \end{bmatrix} x^i_t + \begin{bmatrix}
        0.5 \\ 1
    \end{bmatrix} u^i_t, \ \forall t \in [T-1], \ i \in \{L,F\},
\end{equation}
and a reference trajectory $r^i_t\in \reals^{2}$ that follow independent dynamics given by $r^i_{t+1} = \Sigma^i r^i_t \quad \forall t \in [T-1],\ i \in \{L,F\}.$
The overall system state consists of the leader's and follower's individual states and reference trajectories, $x_t = \big[{x^L_t}^\top, {x^F_t}^\top, {r^L_t}^\top, {r^F_t}^\top\big]^\top$.
The leader's objective $J^L_t$ is to simultaneously avoid the follower while tracking its reference trajectory. Its state cost matrix $Q^L$ is given by
\begin{equation}
    Q^L = \begin{bmatrix}
        Q^{L1} & Q^{L2} & Q^{L3} & 0\\ {Q^{L2}}^\top & \varepsilon I_{2\times2} & 0& 0\\
        {Q^{L3}}^\top & 0 & Q^{L4} &0\\
        0 & 0 & 0 & 0
    \end{bmatrix} \in \mb{S}^{8}_{+},
\end{equation}
where $Q^{L2} \succ 0$, $Q^{L3} \prec 0$, $\varepsilon > 0$.\\
We consider three possible follower intentions: (T) tracking the leader, (I) indifferent to the leader, or (A) avoiding the leader. The state cost matrix $Q^F$ for each intention is
\begin{equation}
\begin{aligned}
    Q^F{(T)} &= \begin{bmatrix}
        Q^{F1}{(T)} & Q^{F3}{(T)} & 0 & 0\\
        {Q^{F3}{(T)}}^\top & Q^{F2}{(T)} & 0 & 0\\
        0 & 0 & 0 & 0\\
        0 & 0 & 0 & 0
    \end{bmatrix} \in \mb{S}^{8}_{+}, \\
        Q^F{(I)} &= \begin{bmatrix}
        0 & 0 & 0 & 0\\
        0 & Q^{F2}{(I)} & 0 & Q^{F4}{(I)}\\
        0 & 0 & 0 & 0\\
        0 & {Q^{F4}{(I)}}^\top & 0 & Q^{F5}{(I)}
    \end{bmatrix} \in \mb{S}^{8}_{+},\\
        Q^F{(A)} &= \begin{bmatrix}
        \alpha I_{2\times2} & Q^{F3}{(A)} & 0 & 0\\
        {Q^{F3}{(A)}}^\top & Q^{F2}{(A)} & 0 & Q^{F4}{(A)}\\
        0 & 0 & 0 & 0\\
        0 & {Q^{F4}{(A)}}^\top & 0 & Q^{F5}{(A)}
    \end{bmatrix} \in \mb{S}^{8}_{+},
\end{aligned}
\end{equation}
where $Q^{F3}{(T)} \prec 0$, $Q^{F4}{(I)} \prec 0$, $Q^{F4}{(A)} \prec 0$ and $\alpha >0$.\\
\textbf{Intention estimation algorithm}. The leader updates its strategies either by assuming one of three fixed follower intentions (T, I, A) or using an adaptive control (Ad). In the adaptive control, the leader updates its assumption about the follower’s BR using a Multiple Model Adaptive Estimator (MMAE)~\cite{estTrackNav}. 

At time $t \in [T]$, the BR belief is $b_t = \argmax_{b} P\big(b|x_t\big)$,
where $P\big(b|x_t\big)$ is the probability of BR $b \in \mc{F}$ based on the current state $x_t \in \reals^n$. The probabilities are initially equal, $P\big(b|x_0\big) = \frac{1}{3}$, and updated according to a Bayesian rule
\begin{equation}
    P\big(b|x_t\big) = \frac{P\big(x_t|b\big) \cdot P(b)}{\sum_{\beta=1}^3 P\big(x_t|\beta\big) \cdot P(\beta)},
\end{equation}
where $P(b) = P\big(b|x_{t-1}\big)$. The model probability for each BR belief $b$ is determined by applying a softmin function to the corresponding residual norm, $\|e_t({b})\|_2$,
with residual $e_t(b) = x_t - \hat{x}_t(b)$, and $\hat{x}_t(b)$ being the propagated dynamics~\eqref{eq:dynamics} if the follower had BR $b$. \\
\textbf{Monte Carlo simulation setup.}
For each metric we run 1000 Monte Carlo simulations with a horizon $T=20$. In each one, the initial state $x_0$ is sampled from a uniform distribution where $\|x_0\|_{\infty}\leq20$ and the reference trajectories are $\Sigma^i = \sigma I^{2\times 2}$, where $0<\sigma\leq1$ is sampled uniformly. 
\subsection{OL information structure}\label{sec:OL_sim_results}
\subsubsection{Optimality against other beliefs} This case study verifies that an incorrect BR belief may lead to obtaining the lowest cost. The following matrix shows for each follower ground truth, what percentage of times each BR belief obtains the lowest cost. 
\begin{figure}[H]
    \centering
\includegraphics[width=1\linewidth]{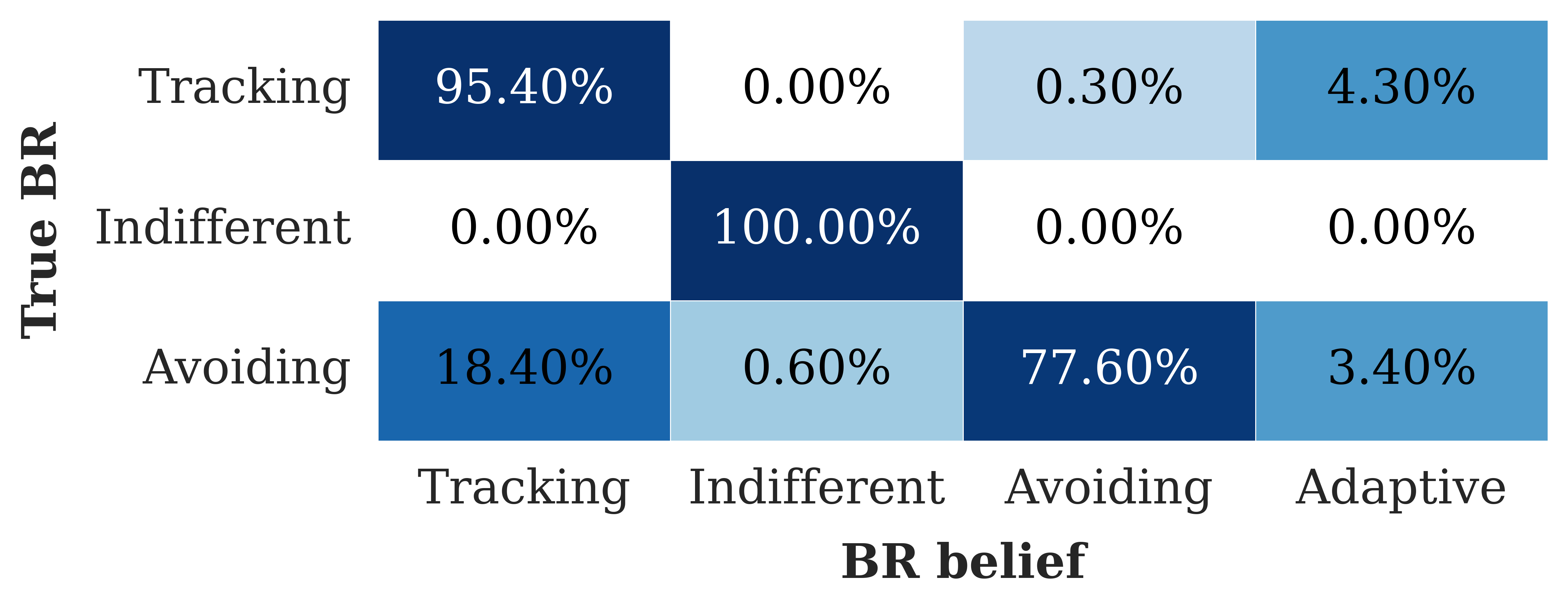}
    \caption{Percent of simulations with lowest cost achieved by each BR belief under OL information structure when $\tau = 1$.}
    \label{fig:OL_conf_mat}
\end{figure}
\subsubsection{Update time case study}
We examine how the update time $\tau$ affects how often each BR belief yields the lowest cost. Figure \ref{fig:OL_update_time} shows this for the case when the follower's true BR is avoiding.
\begin{figure}[H]
    \centering
    \includegraphics[width=1\linewidth]{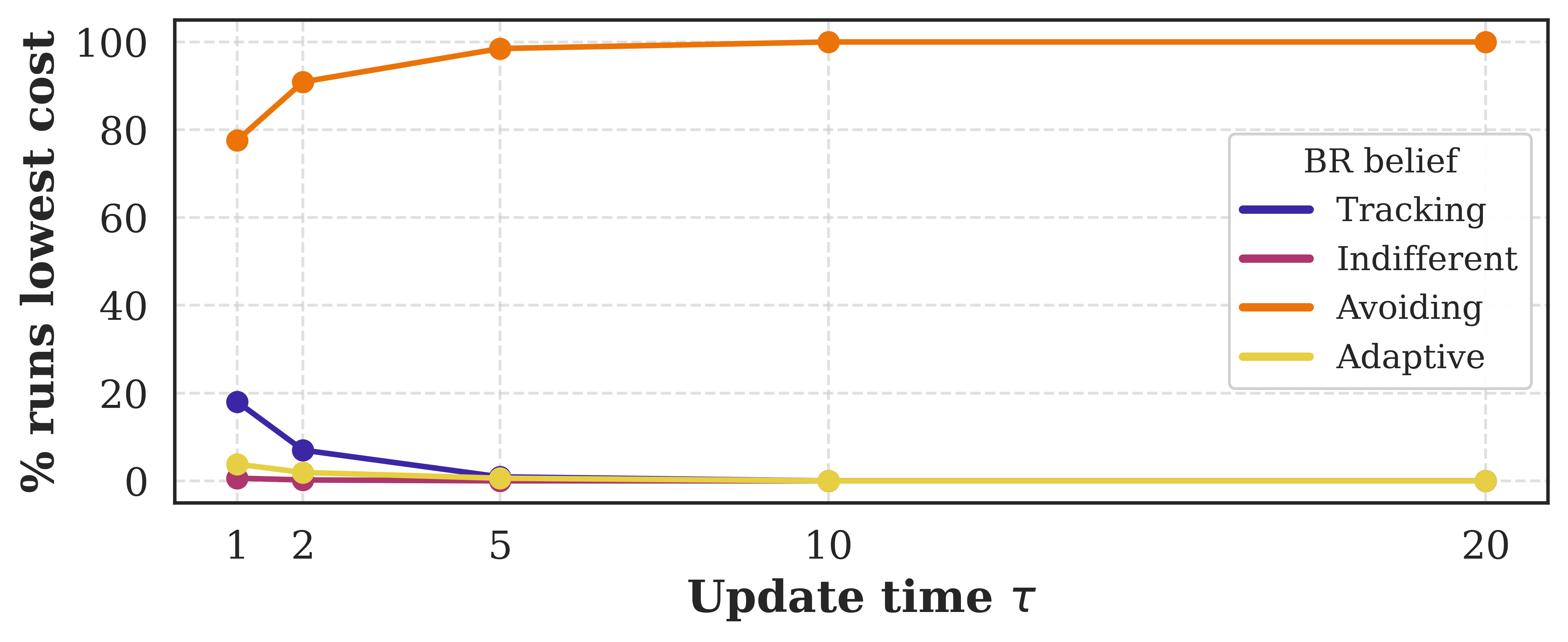}
    \caption{Percentage of simulations where each BR belief obtains the lowest cost for $\tau = 1, 2, 5, 10, 20$ under OL information structure.}
    \label{fig:OL_update_time}
\end{figure}
\subsection{Feedback information structure}\label{sec:FB_sim_results} 
\subsubsection{Optimality against other beliefs}
Identical to Figure~\ref{fig:OL_conf_mat} in setup, but for FSE
(App.~\ref{app:lq_params}).
\begin{figure}[H]
    \centering
    \includegraphics[width=1\linewidth]{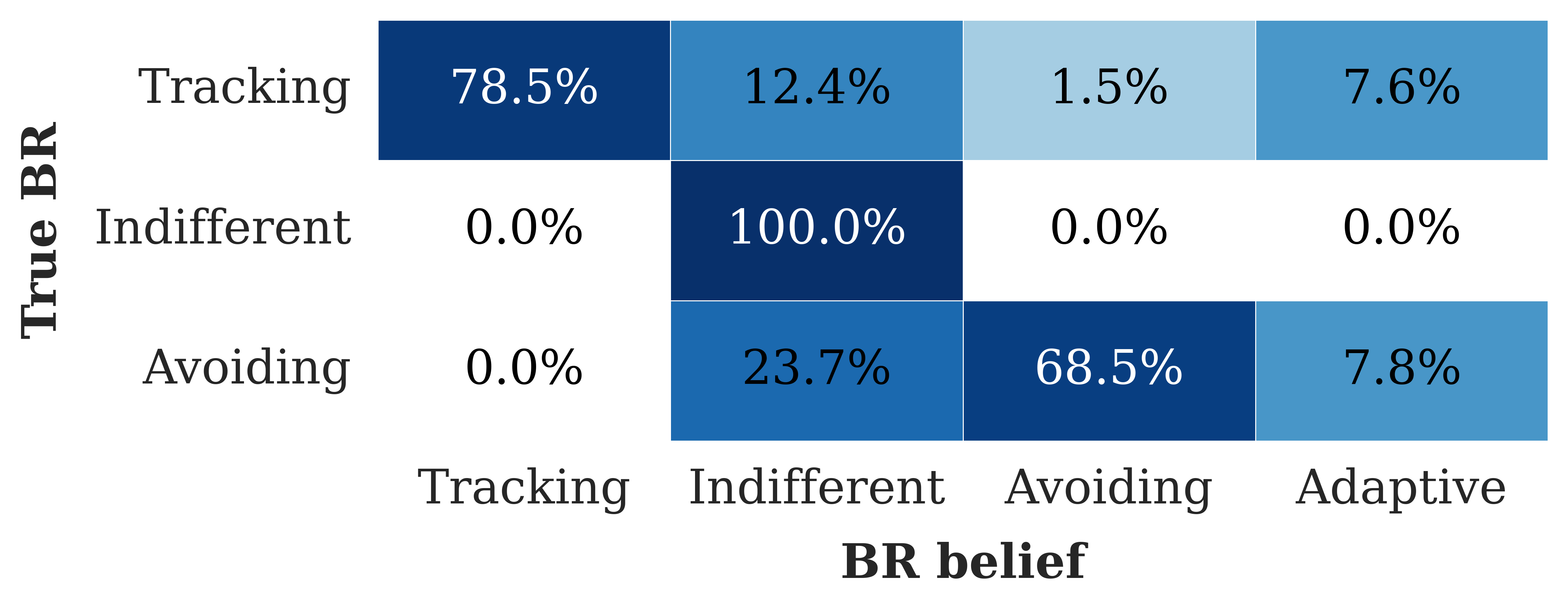}
    \caption{Percent of simulations with lowest cost achieved by BR beliefs under FB information structure and $\tau=1$.}
    \label{fig:FB_confusion_mat}
\end{figure}
\subsubsection{Update time case study}Identical to Figure~\ref{fig:OL_update_time} in setup but for FSE 
(App.~\ref{app:lq_params}).
\begin{figure}[H]
    \centering
    \includegraphics[width=1\linewidth]{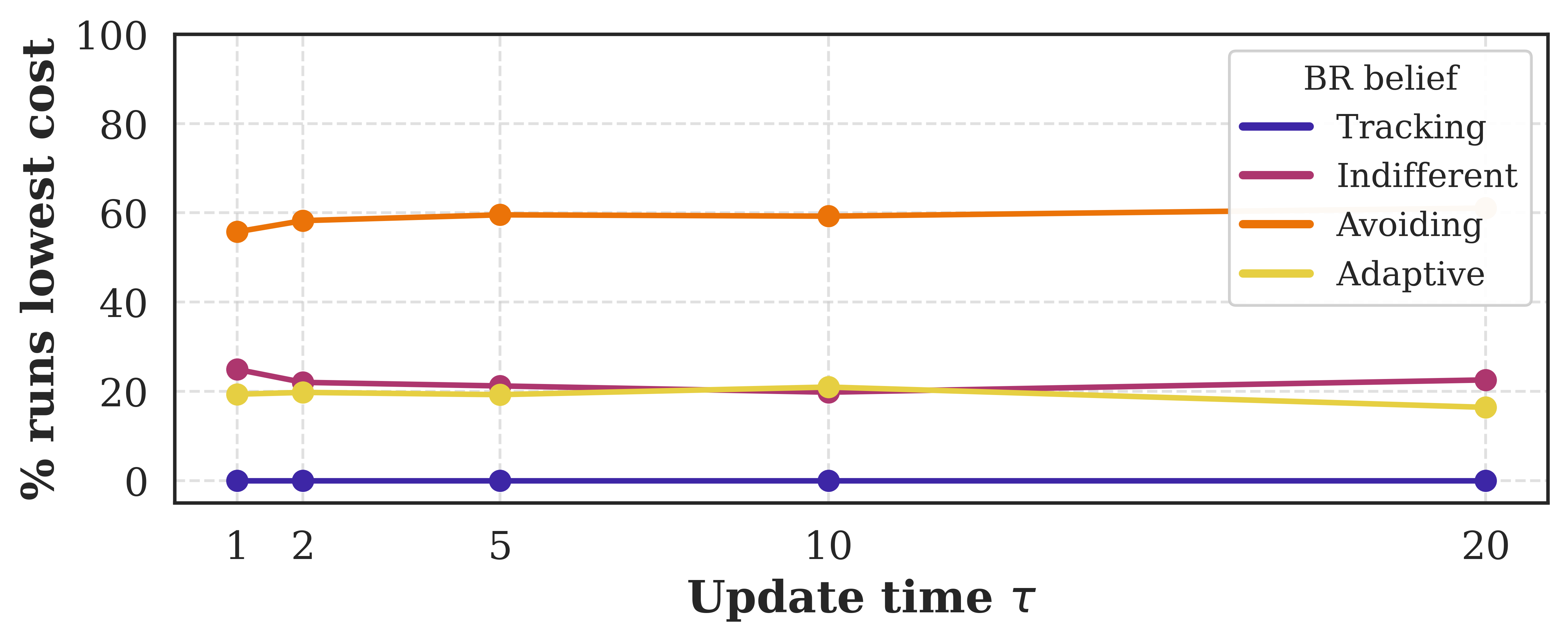}
    \caption{Percentage of simulations where each BR belief obtains the lowest cost for $\tau = 1, 2, 5, 10, 20$ under FB information structure.}
    \label{fig:FB_update_time}
\end{figure}
\subsection{Cost difference}
For the same setting as Figures~\ref{fig:OL_conf_mat} and~\ref{fig:FB_confusion_mat}, we provide the average \% of the cost by which each BR belief exceeds the minimum cost in each run. 
\begin{figure}[H]
    \centering
    \includegraphics[width=1\linewidth]{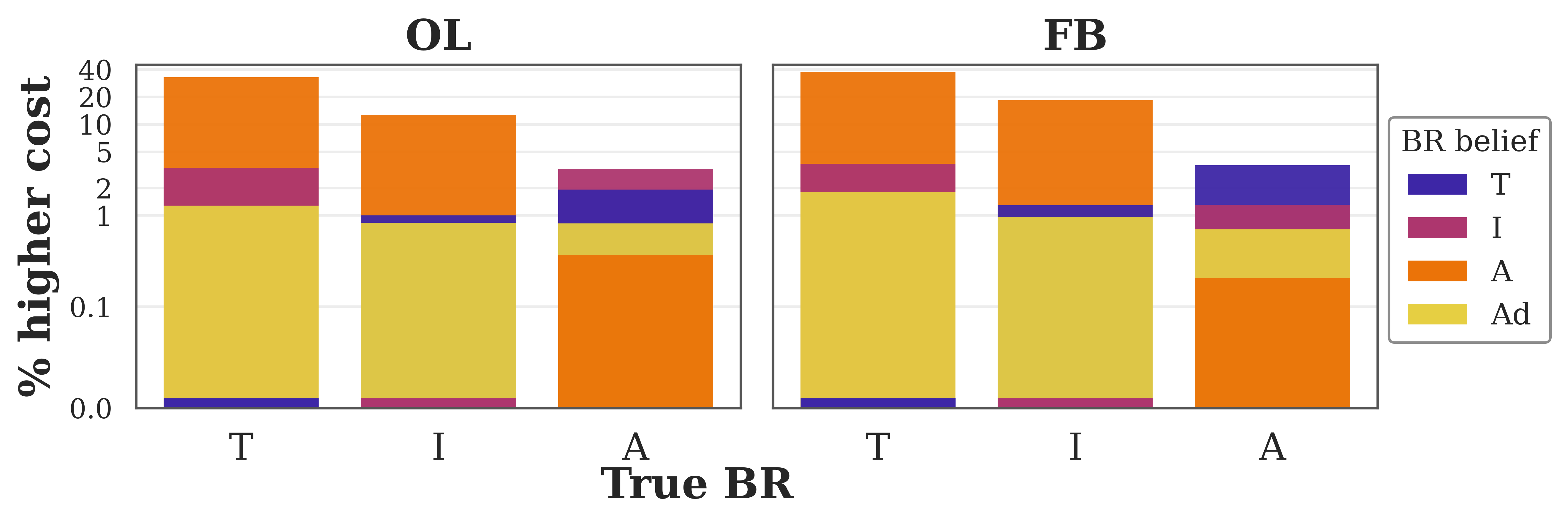}
    \caption{Average percentage by which each belief’s cost exceeds the best model’s cost for each true BR and $\tau=1$ under OL and FB information structures.}
    \label{fig:cost}
\end{figure}
\subsection{Discussion}
The first key observation is that Figures~\ref{fig:OL_conf_mat} and~\ref{fig:FB_confusion_mat} show that the lowest cost is not always achieved by the true BR belief, verifying our results from Sections~\ref{sec:optimality_conditions_OLSE_intention} and~\ref{sec:feedback_opt_break}.
Under both information structures, when the follower is either tracking or avoiding the leader, some alternative beliefs occasionally yield better performance.

Another insight from Figures~\ref{fig:OL_conf_mat} and~\ref{fig:FB_confusion_mat} is that the performance associated with assuming a particular belief depends on the set of beliefs under consideration. 
For example, when the follower is indifferent, assuming the true BR consistently results in the lowest cost. 
However, expanding the set of possible behaviors may alter this outcome. 
This aspect is a subject for future investigation.

Figures~\ref{fig:OL_update_time} and~\ref{fig:FB_update_time} explore the relationship between the update time $\tau$ and the frequency with which the true BR belief is outperformed by others. 
In the OL case, we observe that more frequent updates lead to more instances in which alternative models perform better, with $\tau = T$ corresponding to the optimality condition of the OLSE~\eqref{eq:optCond_OL_intention}. 
In contrast, in the FB case, the number of times each model achieves the lowest cost remains unaffected by $\tau$, which aligns with the time-consistency property of the FSE.

Finally, Figure~\ref{fig:cost} shows that, across all simulations, the true BR belief remains closest to the minimum cost overall, while the cost gap with incorrect BR beliefs is, on average, quite large. 
This finding shows that although some beliefs occasionally outperform the true BR, this does not translate into comparable overall cost levels. 
In all cases, the adaptive solution also performs near-optimally on average.

\section{Conclusion} \label{sec:conclusion}
This paper examines the optimality conditions of Stackelberg OL and FB equilibria when the leader's belief about the follower's BR is updated over the game horizon. We prove that for both OL and FB information structures an incorrect BR belief may result in lower cumulative cost. We illustrate that these inconsistencies arise in a LQ game modeling a collision avoidance scenario. Future work will focus on characterizing the BR beliefs that may yield a lower cost than the true BR and identifying the conditions under which such outcomes occur.

\bibliographystyle{IEEEtran}
\bibliography{citations}
\appendix
\subsection{LQ game}
\label{app:lq_params}
We consider a finite horizon dynamic LQ game. The cost functions of both players is denoted as 
\begin{equation}
    J^i_{0:T} = x_T^\top Q^i x_T + \sum_{t=0}^{T-1} x_t^\top Q^i x_t + {u^i_t}^\top R^i u^i_t, \, i \in \{L,F\},
    \label{eq:LQ_cost_func}
\end{equation}
where $Q^i \succeq 0$ and $R^i \succ 0$ to guarantee strict convexity of the cost functions. We derive the trajectory-level cost matrices as $\bar{Q}^i = \text{blkdiag}(Q^i, Q^i, \ldots, Q^i) \in \reals^{n_i(T+1)\times n_i(T+1)}$ and $\bar{R}^i = \text{blkdiag}(R^i, R^i, \dots, R^i) \in \reals^{m_i T\times m_i T}$ for $i\in \{L, F\}$. The different follower's BR are associated with different $Q^F, R^F$ matrices, defined during usage.
The state trajectory becomes~\eqref{eq:state_trajectory} where the matrices $H \in \reals^{n(T+1)\times n}$, $G^i \in \reals^{n(T+1) \times m_iT}$ are given by
\begin{equation}
        H = \begin{bmatrix}
    I \\ A \\ A^2\\
    \vdots \\
     A^{T}
    \end{bmatrix}, G^i = \begin{bmatrix}
    0 & 0 & \cdots & 0 \\
    B^i & 0 & \cdots & 0 \\
    A B^i & B^i & \cdots & 0 \\
    \vdots & \vdots & \ddots & \vdots \\
    A^{T-1}B^i & A^{T-2}B^i & \cdots & B^i
    \end{bmatrix}.
    \label{eq:H_G_matrix}
\end{equation}

Under OL information structure~\eqref{eq:information_structures} and given $\traj{u}^L$, the follower's BR  $\traj{u}^F = b^\star(\traj{u}^L)$ is given by
\begin{equation}\label{eqn:br_ol_lq}
\begin{aligned}
    b^\star(\traj{u}^L) &= - \big({G^F}^\top \bar{Q}^F G^F + \bar{R}^F\big)^{-1}{G^F}^\top \bar{Q}^F \big( H x_0 + G^L \traj{u}^L\big)\\
    &= \hat{G}^F \traj{u}^L + \hat{H}^F x_0,   
\end{aligned}
\end{equation}
which defines the matrices $\hat{G}^F$ and $\hat{H}^F$. Under $b^\star$ and $\traj{u}^F$, the leader's OLSE strategy is 
\begin{multline}
     \traj{u}^L =- \left[( G^F \hat{G}^F + G^L)^\top \bar{Q}^L (G^F \hat{G}^F + G^L)+ \bar{R}^L \right]^{-1}\\
     ( G^F \hat{G}^F + G^L)^\top \bar{Q}^L(H+G^F \hat{H}^F) x_0.       
\end{multline}
Under FB information structure~\eqref{eq:information_structures}, we first define each player's value function before defining the FSE. For player $i$ at time $t$, we recursively define its value function as
\begin{equation}
    V_t^L = \begin{cases}
        Q^L & t = T,\\
         Q^L + P_t^\top V^L_{t+1} P_t + (K_t^F)^\top R^L K_t^F  & t \in [T-1],
    \end{cases}
\end{equation}
\begin{equation}
    V_t^F = \begin{cases}
        Q^F & t = T,\\
        Q^F + P_t^\top V^F_{t+1} P_t + U_t^\top R^F U_t & t \in [T-1],
    \end{cases}
\end{equation}
where  $P_t = (I - B^F K^F_t)(A - B^L K^L_t)$ and $U_t = K^F_t (A - B^L K^L_t)$ for all $t \in [T-1]$. Using the one-step-ahead value functions $V^i_{t+1}$, the leader and follower's FSE strategies at time step $t$ is given by
\begin{equation}
    \begin{aligned}
        K^F_t =& \big(R^{F} + {B^F}^\top V^F_{t+1} B^F\big)^{-1}  {B^F}^\top V^F_{t+1},\\
        K^L_t=& \Big( {B^L}^\top \big(I - B^F K^F_t\big)^\top V^L_{t+1} \big(I - B^F K^F_t\big)B^L+ R^{L} \Big)^{-1} \\
        &\Big( {B^L}^\top \big(I - B^F K^F_t\big)^\top
    V^L_{t+1}\big(I - B^F K^F_t\big) A\Big),\\
    &\forall t \in [T-1].  
    \end{aligned}
\end{equation}  

\subsection{Numerical OLSE example}\label{app:example1}
Consider a LQ Stackelberg game (App.~\ref{app:lq_params}) with scalar state, scalar leader and follower controls, and OL information structure~\eqref{eq:information_structures}.
The state dynamics is
\begin{equation}
    x_{t+1} = 1.7 x_t + 1.4 u^L_t + 0.5 u^F_t, \quad \forall t\in [T-1].
\end{equation}
The time horizon is $T=5$, update time is $\tau = 3$, and the initial state is $x_0 = 7.6$. The leader's cost function~\eqref{eq:LQ_cost_func}  
has parameters $(Q^L, R^L) = (16, 17)$ and the follower's cost function~\eqref{eq:LQ_cost_func}
with true BR, $b^\star$, is given by $(Q^F,R^F) = (7, 19)$. Let $b'\in \mc{F}$ denote the follower's BR if $(Q^F, R^F) = (8, 9)$. We assume that $b^1=b^2=b'$. 

The leader first calculates the strategies $\traj{u}^{L\star}$ and $\traj{u}^{L'}$ under beliefs $b^\star$ and $b'$, given by
\begin{align}
    \traj{u}^{L\star} &= \left[-4.6, -0.86, 0.06, 0.19, 0.12 \right], \label{eq:uL_calc_0}\\
    \traj{u}^{L'} &= \left[-2.28, 0.4, 0.73, 0.52, 0.25 \right].
\end{align}
The follower's BR is given by
\begin{align}
 b^\star\big(\traj{u}^{L\star}\big) 
&= \left[-7.53, -4.14, -2.29,-1.21, -0.53 \right],\\
b^\star\big(\traj{u}^{L'}\big) 
&= \left[-13.43, -7.57, -4.25,-2.27, -0.98 \right].
\end{align}
At update time $\tau$, the leader's cost under both BR beliefs is 
\begin{equation}
    J^L_{0:2}(b^\star)~=~1443.18, \quad J^L_{0:2}(b')~=~1227.72.
\end{equation} 
Thus, $J^L_{0:2}(b')<J^L_{0:2}(b^\star)$, this shows precisely an example where the optimal solution over the horizon $[T]$ is not optimal over a shorter horizon, $[0, \tau)$. \\
At $t=3$ the state under BR beliefs $b^\star$ is $x^\star_3~=~1.22$ and the state under BR belief $b'$ is $x'_3~=~2.13$. The updated leader controls are $\hat{u}^{L\star}_{3:4}$ and $\hat{u}^{L'}_{3:4}$  given by 
\begin{align}
    \hat{u}^{L\star}_{3:4} = \left[-1.06,-0.3 \right],\label{eq:uL_calc_3} \quad \hat{u}^{L'}_{3:4} = \left[-1.56,-0.41\right].
\end{align} 
We verify that the OLSE of this problem is time-inconsistent since $u^{L\star}_{3:4}~\neq~\hat{u}^{L\star}_{3:4}$. The follower best responds to the leader with
\begin{align}
b^\star\big(\hat{u}^{L\star}_{3:4}\big) = \left[-0.21, -0.07 \right],\quad
    b^\star\big(\hat{u}^{L'}_{3:4}\big) = \left[-0.6,-0.23\right].
\end{align}
Finally, the horizon $[3,5]$ incurs cost under belief $b^\star$ and $b'$
\begin{equation}
    J^L_{3:5}(b^\star)=50.72, \quad J^L_{3:5}(b')=162.88.
\end{equation}
Despite the cost for the horizon $[3,5]$ being greater for $b'$, this belief still obtains lower total cost over $[T]$ since $J^L_{0:5}(b^\star)~=~1493.9$ and $J^L_{0:5}(b')~=~1390.6$.

\subsection{Numerical FSE example}\label{app:example2}
This example presents a LQ Stackelberg game (App.~\ref{app:lq_params}) 
with FB information structure~\eqref{eq:information_structures} where assuming some $b'\in\mc{F}$ achieves a lower cost than $b^\star\in\mc{F}$. Consider the scalar game dynamics
\begin{equation}
    x_{t+1} = 1.4 x_t + 1.7 u^L_t + 1.7 u^F_t, \quad \forall t \in [T-1],
\end{equation}
with a horizon $T=8$, update time $\tau = 3$, and initial condition $x_0 = -5.6$. The leader's cost~\eqref{eq:LQ_cost_func}
is defined by the matrices $(Q^L, R^L) = (7, 16)$. The follower's true BR $b^\star$ is given by $(Q^F, R^F) = (4, 24)$. We consider the case where the leader also assumes a BR $b^1=b^2=b'$ defined by $(Q^F, R^F) = (29, 12)$. Because the FSE is time-consistent and the beliefs remain constant, the leader strategy under beliefs $b^\star$ and $b'$, respectively, are $\traj{K}^{L\star}$ and $\traj{K}^{L'}$ and the follower's BR in either case is $\traj{K}^{F\star}$ and $\traj{K}^{F'}$,
all control numerical values are given by
\begin{align}
    \traj{K}^{L\star} =& [0.31, 0.31, 0.31, 0.31,0.32,
    0.32, 0.33, 0.3],\\
    \traj{K}^{L'} =& [0.01, 0.01, 0.01, 0.01, 0.01, 0.01, 0.01, 0.02],\\ 
    \traj{K}^{F\star} =& [0.27, 0.27, 0.27, 0.26, 0.26,0.26, 0.25, 0.19],\\
    \traj{K}^{F'} =& [0.37, 0.37, 0.37, 0.37, 0.36, 0.35, 0.31, 0.19]. 
\end{align}
After playing the different strategies, the leader observes costs
\begin{equation}
J^L_{0:7}\left(b^\star\right) = 347.2, \quad J^L_{0:7}\left(b' \right) = 299.4.
\end{equation}
Which shows that the case where the leader assumed the wrong follower BR obtained a lower cost than the true BR, verifying Proposition~\ref{prop:feedback_prob1}.

\end{document}